\documentclass[10pt,a4paper]{article}

\usepackage[T1]{fontenc}
\usepackage[english,brazil]{babel}
\usepackage{lmodern}
\usepackage[top=20mm,bottom=20mm,left=20mm,right=20mm,heightrounded]{geometry}
\usepackage{amsmath,amssymb,amsfonts}
\usepackage{amsthm}
\usepackage{booktabs}
\usepackage{mathtools}
\usepackage{caption}
\usepackage{multicol}
\usepackage{microtype}
\usepackage{hyperref}
\usepackage{doi}
\usepackage{url}

\hypersetup{hidelinks}

\newtheorem{proposition}{Proposition}
\newtheorem{corollary}{Corollary}
\newtheorem{remark}{Remark}

\begin{document}

\title{\vspace{-1.5em}
\textbf{A Clarifying Note on the Position--Momentum Correspondence:\
Pontryagin Duality, Fourier Transport, and Physical Normalization}}
	\author{
	Samuel B. Soltau\thanks{Instituição 1, e-mail: autor1@exemplo.br}\\
	\small Department of Physics, Federal University of Alfenas -- Minas Gerais, Brazil
}
\date{September, 2026}
\maketitle

\begin{abstract}
\small
\noindent
The position--momentum correspondence combines several mathematically distinct
identifications that are often conflated. For the additive group $(\mathbb R,+)$,
every continuous character has the form $x\mapsto \mathrm e^{ikx}$. After choosing
a coordinate $p=\alpha k$ on the dual group, one obtains
$\mathrm e^{ipx/\alpha}$, so Pontryagin duality alone does not identify a
particular dual coordinate with physical momentum. For the corresponding unitary
Fourier transform, multiplication in position space is transported to convolution
in momentum space, while the diagonal distribution is transported to the specific
composition with the addition map $(p,q)\mapsto p+q$; it is not transported to
pointwise multiplication in momentum space. The family
$\hat P_\alpha=-i\alpha\,d/dx$ has commutator
$[\hat X,\hat P_\alpha]=i\alpha I$ on $\mathcal S(\mathbb R)$, and the standard
quantum-mechanical normalization $[\hat X,\hat P]=i\hbar I$ therefore selects
$\alpha=\hbar$ within this family. Equivalently, the standard normalization of
the translation generator gives $\hat P=-i\hbar\,d/dx$. The resulting characters
$\mathrm e^{ipx/\hbar}$ have spatial period $h/|p|$ for $p\neq0$. All
distributional statements are formulated in the Schwartz rigging and no product or
pullback of arbitrary tempered distributions is used.
\end{abstract}


\vspace{0.7em}

\begin{multicols}{2}
\raggedcolumns

\section{Introduction}
\label{sec:introduction}

The standard position--momentum correspondence can be described simultaneously
in harmonic analysis, distribution theory, and quantum mechanics. In the position
representation, momentum is represented by a differential operator, whereas in
the momentum representation it is represented by multiplication. The Fourier
transform relates these representations.

The familiar relations
\begin{equation}
 p=\hbar k,
 \qquad
 \lambda=\frac{h}{|p|},
 \qquad
 h=2\pi\hbar
 \label{eq:intro_debroglie}
\end{equation}
are standard physical relations and are not claimed here as new physical results
\cite{deBroglie1923,deBroglie1925}.

The purpose of this note is narrower and structural: to separate the logical
ingredients entering these relations. Four levels should be distinguished. First,
Pontryagin duality identifies the character group of the translation group.
Second, after choosing Haar measures and a Fourier convention, Fourier analysis
realizes the corresponding pairing. Third, the unitary representation of spatial
translations has a self-adjoint generator, and Fourier transformation gives its
spectral representation. Fourth, a physical convention fixes which normalization
of that generator is to be called momentum.

No claim is made that Planck's constant can be derived from the abstract dual group.
The point is instead that the mathematical structure contains a normalization
freedom which is subsequently calibrated by the physical definition of momentum.
This is also the reason why the formula $p=\hbar k$ should not be attributed to
Pontryagin duality alone.

The continuous-spectrum aspects require some care. Plane waves and Dirac
measures/distributions are not vectors in $L^2(\mathbb R)$, and multiplication of
arbitrary tempered distributions is not defined in general. We therefore work in
the standard rigged Hilbert space
$\mathcal S(\mathbb R)\subset L^2(\mathbb R)\subset\mathcal S'(\mathbb R)$
and use only explicitly defined operations.

\section{Characters of the Translation Group}
\label{sec:duality}

Let
\[
G=(\mathbb R,+).
\]
Its Pontryagin dual is
\begin{equation}
 \widehat G=
 \operatorname{Hom}_{\mathrm{cont}}(G,U(1)).
 \label{eq:dual_group}
\end{equation}
Every continuous character is of the form
\begin{equation}
 \chi_k(x)=\mathrm e^{ikx},
 \qquad k\in\mathbb R,
 \label{eq:character_k}
\end{equation}
so, after fixing this parameterization, $\widehat G$ is identified with another
copy of $(\mathbb R,+)$ \cite{Folland1995,Folland1999}.

The quantity $k$ is the wave-number coordinate when $x$ has dimensions of length.
Introduce a further coordinate $p$ by
\begin{equation}
 p=\alpha k,
 \qquad \alpha>0.
 \label{eq:p_alpha_k}
\end{equation}
The same character is then written as
\begin{equation}
 \chi^{(\alpha)}_p(x)=\mathrm e^{ipx/\alpha}.
 \label{eq:character}
\end{equation}
Thus the abstract dual group, by itself, does not single out a distinguished
physical unit or scale for a coordinate named $p$.

Indeed, for $c>0$, the simultaneous rescaling
\[
p'=cp,
\qquad
\alpha'=c\alpha
\]
preserves $p/\alpha$ and hence leaves the character invariant:
\begin{equation}
 \chi^{(\alpha')}_{p'}(x)=\chi^{(\alpha)}_p(x).
 \label{eq:scaling_invariance}
\end{equation}

\begin{proposition}[Coordinate-scale freedom]
\label{prop:scaling}
Within the parameterization $p=\alpha k$, the normalization of the chosen
coordinate on $\widehat G$ is not fixed by the abstract character group. It is
invariant under
\[
p\mapsto cp,
\qquad
\alpha\mapsto c\alpha,
\qquad c>0.
\]
\end{proposition}

\begin{proof}
The character depends only on the dimensionless combination $px/\alpha$, which is
invariant under the stated simultaneous rescaling.
\end{proof}

This is a statement about coordinate choice, not a defect of Pontryagin duality.
Once the pairing, Haar measures, and coordinate convention are fixed, the Fourier
analysis is unambiguous.

\section{Fourier Analysis and Operator Intertwining}
\label{sec:fourier}

\subsection{Schwartz setting and Fourier normalization}

Let $\mathcal S(\mathbb R)$ denote the Schwartz space and $\mathcal S'(\mathbb R)$
its continuous dual. We use the rigged Hilbert-space chain
\begin{equation}
 \mathcal S(\mathbb R)
 \subset
 L^2(\mathbb R,dx)
 \subset
 \mathcal S'(\mathbb R).
 \label{eq:rhs}
\end{equation}
The Fourier transform preserves $\mathcal S(\mathbb R)$ and extends by continuity
to a unitary operator on $L^2(\mathbb R)$ \cite{Folland1999,Hormander1990}.

For $\alpha>0$, define
\begin{equation}
 (\mathcal F_\alpha f)(p)
 =
 \frac{1}{\sqrt{2\pi\alpha}}
 \int_{\mathbb R}
 \mathrm e^{-ipx/\alpha}f(x)\,dx.
 \label{eq:Falpha}
\end{equation}
The inverse is
\begin{equation}
 f(x)
 =
 \frac{1}{\sqrt{2\pi\alpha}}
 \int_{\mathbb R}
 \mathrm e^{ipx/\alpha}(\mathcal F_\alpha f)(p)\,dp.
 \label{eq:Falpha_inverse}
\end{equation}
With Lebesgue measure on both axes, $\mathcal F_\alpha$ is unitary from
$L^2(\mathbb R,dx)$ to $L^2(\mathbb R,dp)$. The prefactor in
\eqref{eq:Falpha} determines the constants in the convolution identities below.

\subsection{Position, momentum, and translation generators}

On $\mathcal S(\mathbb R)$ define
\begin{equation}
 (\hat Xf)(x)=xf(x),
 \qquad
 (\hat P_\alpha f)(x)=-i\alpha f'(x).
 \label{eq:XPalpha}
\end{equation}
A direct calculation gives
\begin{equation}
 [\hat X,\hat P_\alpha]f=i\alpha f,
 \qquad f\in\mathcal S(\mathbb R).
 \label{eq:CCRalpha}
\end{equation}
Integration by parts yields
\begin{align}
 \mathcal F_\alpha\hat P_\alpha\mathcal F_\alpha^{-1}
 &=p,
 \label{eq:P_intertwine}
 \\
 \mathcal F_\alpha\hat X\mathcal F_\alpha^{-1}
 &=i\alpha\frac{d}{dp},
 \label{eq:X_intertwine}
\end{align}
where the first right-hand side denotes multiplication by $p$.

The connection with translations can be stated without identifying the scale yet.
Let
\begin{equation}
 (U(a)f)(x)=f(x-a),
 \qquad a\in\mathbb R.
 \label{eq:translation}
\end{equation}
Then $U(a)$ is a strongly continuous unitary representation of $(\mathbb R,+)$.
By Stone's theorem there is a unique self-adjoint generator $A$ such that
$U(a)=\exp(-iaA)$. With
\[
A_\alpha=\frac{\overline{\hat P_\alpha}}{\alpha},
\]
we have
\begin{equation}
 U(a)=\exp\!\left(-\frac{ia}{\alpha}\overline{\hat P_\alpha}\right),
 \label{eq:U_Palpha}
\end{equation}
and, in the Fourier representation,
\begin{equation}
 (\mathcal F_\alpha U(a)\mathcal F_\alpha^{-1}u)(p)
 =\mathrm e^{-iap/\alpha}u(p).
 \label{eq:translation_spectral}
\end{equation}
Thus the Fourier variable $p$ is the spectral coordinate of the normalized
translation generator $A_\alpha$; whether that coordinate is called physical
momentum is an additional normalization convention \cite{Hall2013}.

\subsection{Position multiplication and convolution}

For $f,g\in\mathcal S(\mathbb R)$ let
\begin{equation}
 m_X(f,g)=fg.
 \label{eq:mX}
\end{equation}
Then
\begin{equation}
 \mathcal F_\alpha(fg)
 =
 \sqrt{2\pi\alpha}\,
 (\mathcal F_\alpha f)*(\mathcal F_\alpha g),
 \label{eq:fourier_product}
\end{equation}
where
\[
(u*v)(p)=\int_{\mathbb R}u(q)v(p-q)\,dq.
\]
Consequently, if
\begin{equation}
 \mu_\alpha(u,v)
 =\sqrt{2\pi\alpha}\,(u*v),
 \label{eq:mu_alpha}
\end{equation}
then
\begin{equation}
 \mathcal F_\alpha\circ m_X
 =
 \mu_\alpha\circ(\mathcal F_\alpha\otimes\mathcal F_\alpha).
 \label{eq:m_fourier}
\end{equation}
In particular,
\begin{equation}
 m_X\longmapsto\mu_\alpha,
 \qquad
 m_X\not\longmapsto m_P,
 \label{eq:not_mP}
\end{equation}
where $m_P(u,v)=uv$ denotes pointwise multiplication in the $p$ coordinate.

\begin{remark}
The factor in \eqref{eq:fourier_product} depends on the Fourier normalization.
The structural statement does not: Fourier transformation converts pointwise
multiplication into convolution rather than into pointwise multiplication in the
dual variable.
\end{remark}

\subsection{The diagonal distribution}

The multiplication map has the associated diagonal restriction
\begin{equation}
 \widetilde m_X:
 \mathcal S(\mathbb R^2)\to\mathcal S(\mathbb R),
 \qquad
 (\widetilde m_XF)(x)=F(x,x).
 \label{eq:linearized_mX}
\end{equation}
Its transpose is defined by weak pairing. In particular, for $f\in\mathcal S(\mathbb R)$,
\begin{equation}
 (\widetilde m_X^{\,t}f)(x_1,x_2)
 =f(x_1)\delta(x_1-x_2).
 \label{eq:transpose_mX}
\end{equation}
We denote this tempered distribution by $\delta_Xf$.
No Hilbert-space adjoint of the bilinear map $m_X$ is being asserted.

Define the linear addition map
\[
\Sigma:\mathbb R^2\to\mathbb R,
\qquad
\Sigma(p,q)=p+q.
\]
For $\varphi\in\mathcal S(\mathbb R)$, the function
$\varphi\circ\Sigma$ is smooth and bounded, with derivatives bounded as well;
hence it defines a tempered distribution on $\mathbb R^2$ by integration against
Schwartz test functions.

A direct calculation gives
\begin{equation}
 (\mathcal F_\alpha\otimes\mathcal F_\alpha)(\delta_X f)
 =
 \frac{1}{\sqrt{2\pi\alpha}}
 \bigl(\mathcal F_\alpha f\bigr)\circ\Sigma.
 \label{eq:delta_fourier}
\end{equation}
Indeed,
\begin{align}
&\bigl[(\mathcal F_\alpha\otimes\mathcal F_\alpha)(\delta_X f)\bigr](p,q)
\nonumber\\
&\quad=
\frac{1}{2\pi\alpha}
\int_{\mathbb R}\int_{\mathbb R}
\mathrm e^{-i(px+qy)/\alpha}f(x)\delta(x-y)\,dx\,dy
\nonumber\\
&\quad=
\frac{1}{2\pi\alpha}
\int_{\mathbb R}
\mathrm e^{-i(p+q)x/\alpha}f(x)\,dx
\nonumber\\
&\quad=
\frac{1}{\sqrt{2\pi\alpha}}
(\mathcal F_\alpha f)(p+q).
\label{eq:delta_fourier_calculation}
\end{align}
No pullback theorem for arbitrary distributions is invoked here: the composition
$\varphi\circ\Sigma$ is formed for the specific Schwartz function
$\varphi=\mathcal F_\alpha f$.

The appearance of $p+q$ is the distributional counterpart of convolution. It is
therefore structurally different from pointwise multiplication in the $p$ variable.

\subsection{Constants and evaluation at the origin}

The constant distribution has Fourier transform
\begin{equation}
 \mathcal F_\alpha[1]
 =\sqrt{2\pi\alpha}\,\delta_0,
 \label{eq:fourier_constant}
\end{equation}
while
\begin{equation}
 \int_{\mathbb R}f(x)\,dx
 =\sqrt{2\pi\alpha}\,(\mathcal F_\alpha f)(0).
 \label{eq:epsilon_fourier}
\end{equation}
Likewise,
\begin{equation}
 \mathcal F_\alpha^{-1}[1]
 =\sqrt{2\pi\alpha}\,\delta_0.
 \label{eq:inverse_constant}
\end{equation}
Thus integration in position space is related by Fourier transformation to
evaluation at the origin in the dual coordinate, and the constant distribution is
transformed into a Dirac distribution rather than another constant distribution.

\begin{remark}
The identities in this section are the distributional statements needed here.
They do not establish a global equivalence of Frobenius structures in a category of
arbitrary tempered distributions. Such a claim would require separate hypotheses
on multiplication, composition, domains, and tensorial contractions.
\end{remark}

\section{Weyl Relations and Physical Normalization}
\label{sec:physical}

Let $\overline{\hat X}$ and $\overline{\hat P_\alpha}$ denote the self-adjoint
closures of the operators in \eqref{eq:XPalpha}; on $\mathcal S(\mathbb R)$ they
are essentially self-adjoint. Define
\begin{equation}
\begin{aligned}
T_\alpha(a)&=
\exp\!\left(-\frac{i}{\alpha}a\overline{\hat P_\alpha}\right),\\
M_\alpha(b)&=
\exp\!\left(\frac{i}{\alpha}b\overline{\hat X}\right),
\end{aligned}
\qquad a,b\in\mathbb R.
\label{eq:weylops}
\end{equation}
Then
\begin{align*}
(T_\alpha(a)f)(x)&=f(x-a),\\
(M_\alpha(b)f)(x)&=\mathrm e^{ibx/\alpha}f(x),
\end{align*}
and hence
\begin{equation}
 T_\alpha(a)M_\alpha(b)
 =
 \mathrm e^{-iab/\alpha}M_\alpha(b)T_\alpha(a).
 \label{eq:weylalpha}
\end{equation}
This is the Weyl relation for the chosen normalization. It is a
representation-theoretic statement and is not identified here with the categorical
notion of strong complementarity \cite{CoeckeDuncanKissingerWang2012}.

\subsection{Physical normalization}

The physical momentum observable is conventionally normalized by the canonical
commutation relation
\begin{equation}
 [\hat X,\hat P]=i\hbar I
 \qquad\text{on }\mathcal S(\mathbb R).
 \label{eq:physical_ccr}
\end{equation}
Comparing with \eqref{eq:CCRalpha} gives, within the family
$\hat P_\alpha=-i\alpha\,d/dx$,
\begin{equation}
 \alpha=\hbar.
 \label{eq:alpha_hbar}
\end{equation}

\begin{proposition}[Physical normalization]
\label{prop:physical}
Fix the position coordinate $x$ and the translation representation
$(U(a)f)(x)=f(x-a)$. Within the family
\[
\hat P_\alpha=-i\alpha\frac{d}{dx},
\]
the normalization
\[
[\hat X,\hat P]=i\hbar I
\]
uniquely selects $\alpha=\hbar$.
\end{proposition}

\begin{proof}
By \eqref{eq:CCRalpha},
\[
[\hat X,\hat P_\alpha]f=i\alpha f
\qquad (f\in\mathcal S(\mathbb R)).
\]
Comparison with \eqref{eq:physical_ccr} gives $\alpha=\hbar$.
\end{proof}

There is an equivalent generator statement. With the physical normalization,
Stone's theorem gives
\begin{equation}
 U(a)=\exp\!\left(-\frac{ia}{\hbar}\overline{\hat P}\right),
 \qquad
 \overline{\hat P}=-i\hbar\frac{d}{dx}
\end{equation}
with its standard self-adjoint realization on $L^2(\mathbb R)$. In the momentum
representation, $\hat P$ is multiplication by $p$. Thus $p$ is the spectral
coordinate of the physically normalized translation generator.

This is a physical calibration of the chosen coordinate on the dual group, not a
derivation of Planck's constant from Pontryagin duality. It also makes explicit
what is fixed by Stone's theorem and what is fixed by physical convention.

With $\alpha=\hbar$, the character becomes
\begin{equation}
 \chi_p(x)=\mathrm e^{ipx/\hbar}.
 \label{eq:physical_character}
\end{equation}
Comparison with $\chi_k(x)=\mathrm e^{ikx}$ gives
\begin{equation}
 k=\frac{p}{\hbar},
 \qquad
 p=\hbar k.
 \label{eq:p_hbar_k}
\end{equation}

\section{de Broglie Wavelength}
\label{sec:debroglie}

For $p\neq0$, a positive period $\lambda$ of
$\mathrm e^{ipx/\hbar}$ satisfies
\[
\mathrm e^{ip(x+\lambda)/\hbar}=\mathrm e^{ipx/\hbar},
\]
so
\[
\frac{p\lambda}{\hbar}=2\pi n
\]
for some nonzero integer $n$. The least positive period is therefore
\begin{equation}
 \lambda=\frac{2\pi\hbar}{|p|}
 =\frac{h}{|p|},
 \qquad p\neq0,
 \label{eq:debroglie}
\end{equation}
where $h=2\pi\hbar$. For $p=0$ the character is constant, and the corresponding
wavelength is naturally interpreted as $\lambda=\infty$.

\begin{corollary}[de Broglie wavelength]
\label{cor:debroglie}
After the physical normalization $\alpha=\hbar$, the generalized momentum
character
\[
\chi_p(x)=\mathrm e^{ipx/\hbar}
\]
has spatial period $h/|p|$ for $p\neq0$.
\end{corollary}

\begin{proof}
Substitution of $k=p/\hbar$ into the basic period $2\pi/|k|$ gives
\eqref{eq:debroglie}.
\end{proof}

The statement concerns an ideal character, i.e. a generalized momentum
eigenfunction. A generic $L^2$ wave packet need not possess a single wavelength;
its Fourier transform generally has support over a range of momenta.

\section{Logical Status and Relation to Complementarity}
\label{sec:scope}

The argument above is deliberately modest. It is a clarification of logical
dependence, not a new derivation of the de Broglie relation.

The dependence may be summarized as
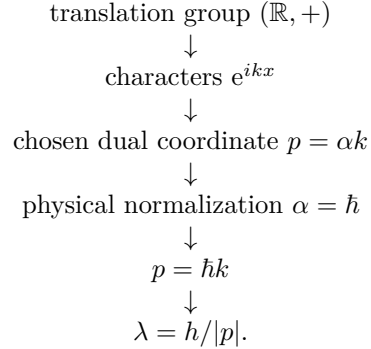
\[
\begin{array}{c}
\text{translation group }(\mathbb R,+)\\
\downarrow\\
\text{characters }\mathrm e^{ikx}\\
\downarrow\\
\text{chosen dual coordinate }p=\alpha k\\
\downarrow\\
\text{physical normalization }\alpha=\hbar\\
\downarrow\\
 p=\hbar k\\
\downarrow\\
 \lambda=h/|p|.
\end{array}
\]
The first three stages are mathematical conventions and constructions; the fourth
is the physical calibration that identifies the chosen dual coordinate with the
standard momentum observable. In particular, Planck's constant is not derived
from the topology of the dual group.

The distributional assertions are formulated through the Schwartz rigging because
Dirac distributions and plane waves are generally not $L^2$ vectors. The operations
used explicitly are the Fourier transform on $\mathcal S$ and $\mathcal S'$, the
transpose of a continuous Schwartz-space map, pointwise multiplication of Schwartz
functions, convolution of Schwartz functions, and the specific composition with
the linear addition map appearing in \eqref{eq:delta_fourier}. No global
multiplication of arbitrary tempered distributions is assumed.

The discussion of Weyl relations is likewise deliberately separated from stronger
notions of complementarity. Strong complementarity in categorical quantum mechanics
is an algebraic condition on interacting observable structures; the Weyl relation
is a relation among a representation of translations and its conjugate modulation
group \cite{CoeckeDuncanKissingerWang2012}. Operational complementarity of quantum
observables is a distinct notion with its own formulations and hypotheses
\cite{KiukasLahtiPellonpaaYlinen2019}.

The main structural claim is therefore limited but precise: the momentum coordinate
used in the usual Fourier representation is a chosen coordinate on the character
group of translations; Fourier transport does not turn position multiplication into
pointwise momentum multiplication; and the physical normalization of the translation
generator fixes the scale so that the standard momentum variable satisfies
$p=\hbar k$.

\section{Conclusion}
\label{sec:conclusion}

For $(\mathbb R,+)$, Pontryagin duality gives the continuous characters
$\mathrm e^{ikx}$. Writing $p=\alpha k$ gives $\mathrm e^{ipx/\alpha}$, with the
scale of the coordinate $p$ not fixed by the abstract character group.

With the unitary Fourier transform \eqref{eq:Falpha}, position multiplication
satisfies
\[
\mathcal F_\alpha(fg)
=
\sqrt{2\pi\alpha}\,(\mathcal F_\alpha f)*(\mathcal F_\alpha g),
\]
so it is not carried into pointwise multiplication in the momentum coordinate.
The diagonal distribution is instead carried to
\[
\frac{1}{\sqrt{2\pi\alpha}}
(\mathcal F_\alpha f)\circ\Sigma,
\qquad
\Sigma(p,q)=p+q.
\]

For the translation representation $U(a)f(x)=f(x-a)$, Stone's theorem supplies
its self-adjoint generator. Imposing the physical normalization
$[\hat X,\hat P]=i\hbar I$ identifies that generator as
$\hat P=-i\hbar\,d/dx$ on the standard core, and the corresponding spectral
coordinate satisfies
\[
p=\hbar k.
\]
The periodicity of the character then gives
\[
\lambda=\frac{h}{|p|}
\]
for $p\neq0$.

Thus the de Broglie relation is not a consequence of Pontryagin duality alone.
The abstract duality supplies the characters; Fourier analysis realizes the pairing
and spectral representation; the physical normalization identifies the dual
coordinate with momentum; and the wavelength follows from the periodicity of the
resulting character. This is the limited claim of the note.

\end{multicols}

\end{document}